\documentclass[a4paper,11pt,leqno]{article}
\usepackage{color}
\usepackage{amsmath}
\catcode`\@=11 \@addtoreset{equation}{section}

\catcode`\@=12

\usepackage{amssymb}
\usepackage{amsfonts}
\usepackage{xcolor}
\usepackage{graphics}
\usepackage{epsfig,psfrag,graphicx}
\usepackage{amssymb}
\usepackage{eepic,epic}
\usepackage{dsfont}

\usepackage{enumerate}

\usepackage{epsfig} 
\usepackage{amsmath} 
\usepackage{amsthm} 
\usepackage{amssymb} 
\usepackage{amsbsy} 
\usepackage{bbm}
\usepackage{verbatim}
\usepackage{stmaryrd}

\usepackage{yfonts}

\usepackage{lmodern}

\theoremstyle{definition}
\newtheorem{definition}{Definition}[section]

\theoremstyle{plain}
\newtheorem{theorem}[definition]{Theorem}
\newtheorem{proposition}[definition]{Proposition}

\newtheorem{remark}[definition]{Remark}

\newcommand{\nc}{\newcommand}
\nc{\weak}{\rightharpoonup}
\nc{\weakstar}{\stackrel{\ast}{\rightharpoonup}}

\nc{\modular}[1]{{\stackrel{ #1}{\longrightarrow\,}}}

\newcommand{\G}{\Gamma}

\allowdisplaybreaks

\DeclareMathOperator*{\bigtimes}{\vartimes}

\title{\LARGE \bf{Characterizations of continuous adequate objective functions for ordinal or interval scaled data}}
\author{ \textsl{Gianni Bosi}$\,^{1}$, \textsl{Gabriele Sbaiz}$\,^{1}$ and \textsl{Magalì Zuanon}$\,^{2}$
 \vspace{.2cm} \\
\footnotesize{$\,^1\;$ \textit{Department of Economics, Business, Mathematics and Statistics}} \\ \footnotesize{\textsc{University of Trieste}} \\
\footnotesize{Via Valerio 4/1, 34127 Trieste, Italy} \vspace{0.1cm} \\
\footnotesize{$\,^2\;$ \textit{Department of Economics and Management}} \\ \footnotesize{\textsc{University of Brescia}} \\
\footnotesize{Via Porta Pile 17B/17C, 25122 Brescia, Italy} \vspace{0.1cm} \\
\vspace{.2cm} \\
\footnotesize{\ttfamily{gianni.bosi@deams.units.it}$\,,\quad$  
\ttfamily{gabriele.sbaiz@deams.units.it}$\,,\quad$  
\ttfamily{magali.zuanon@unibs.it}} \vspace{.1cm}
}

\date{\small \today}

\begin{document}
\maketitle

\abstract{Objective functions (goodness criteria which have to be optimized) that are considered, for instance, in cluster analysis, factor analysis, (linear) structural equation modeling, (linear) regression, multidimensional scaling, choice theory, and utility theory, must be {\em adequate}, i.e. carefully adapted to the structure of the observed data. Adequateness of an objective function means, in our terminology, that (certain) transformations  of its arguments, e.g. changes in the unit measures of the quantities involved, do not influence the solutions of the optimization procedure. In this paper we concentrate our attention on affine strictly increasing transformations and we also incorporate the case of continuous adequate objective functions accordingly. The characterization of adequate dissimilarity coefficients for interval scaled data shows the appropriateness of the concept of adequateness that is developed in this paper.
}



\section{Introduction}

Given a set $\Omega$ describing certain conditions, each value $X(\omega )$ ($\omega \in \Omega$) is measured (observed) on a scale which, for instance, states its
reading in acres or hectares, miles or meters, gills or liters, grains or grams, Fahrenheit or
Celsius or on some multistage rating scale measuring psychological phenomena. Of course, any analysis of measured (observed) data
has to be independent of individual readings of the measured (observed) data. Different
readings means that the underlying units of measurement may differ, i.e. acres instead of
hectares, miles instead of meters, gills instead of liters, grains instead of grams, Fahrenheit
instead of Celsius, one multistage rating scale instead of another one. But under no circumstances the
arbitrarily chosen units of measurement should influence the analysis of data. This means
precisely that the measured (observed) values $X(\omega )$ must be assumed to be only known up
to those transformations that convert one unit of measurement into another one or one multistage rating
scale into another multistage rating scale. This postulate of independence is fundamental for any analysis
of data. It means, on the one hand, that each value $X(\omega )$ only is known up to some group
$ ({\bf T}_{\omega}, \circ ) $ which consists of transformations (bijective functions) ${ T}_{\omega}: {\mathbb R} \mapsto {\mathbb R}$ (converters) that
represent the particular degree of measurement of the datum measured by
$X(\omega )$ and, on the other hand, that any analysis of measured (observed) data $(X(\omega ))_{\omega \in \Omega}$, for
every $X \in {\mathcal X}$, must be compatible with a sub-group $({\bf T}, \circ )$ of the direct product
$({\bf T}^{\times}, \circ ):=(\displaystyle{\bigtimes_{\omega \in \Omega}}{\bf T}_{\omega } , \circ )$ of the groups $ ({\bf T}_{\omega}, \circ ) $
that exactly represents the degree of measurement
(the nature) of the observed data. Analyzing data by objective functions (goodness criteria
which have to be optimized), as frequently being applied in cluster analysis, factor analysis, (linear) structural equation modeling, (linear) regression, multidimensional scaling, choice
theory, utility theory (cf. the Introduction in \cite{Herd5}), therefore, must be
\textit{adequate}, i.e. carefully adapted to the above described possible structure of the measured
(observed) data (cf. Section 1.7 in \cite{Herd5}). This concept of \textit{adequateness} is in
the focus of the following considerations. To be precise, in our terminology an objective function is said to be {\em adequate} with respect to a certain group transformation if its minimization leads to the same result irrespectively of the consideration of any transformation belonging to the group considered. This is precisely the definition of adequateness as it was presented and developed by \cite{Herd5}. The power of the approach based on the identification of suitable objective functions was demonstrated by \cite{Herd4,Herd5}, who also presented a lot of examples concerning Statistics and Social Sciences.

One of the most fundamental problems in (multivariate) data analysis,
therefore, is (at least in opinion of the authors) determining (characterizing) for any fixed given sub-group $({\bf T}, \circ)$ of $(\displaystyle{\bigtimes_{\omega \in \Omega}} {\bf T}_{\omega}, \circ)$ all objective functions that are adequate with respect to $({\bf T}, \circ)$.

This problem is rarely discussed in the literature. Indeed, with the
exception of the papers by \cite{Herd1,Herd2,Herd3}, \cite{Herd4,Herd5}, and \cite{Sar} that at least
somewhat touch the general adequateness problem, no other paper has been found by the
authors that, actually, discusses this topic.

In this paper, we are primarily concerned with strictly increasing affine transformation. We prove a very important and restrictive result, which characterizes adequateness of nonnegative objective functions with respect to such transformation group (see Theorem \ref{Theochar}). Such important characterization has also a continuous version, when only bounded function representing different measurements are considered. The case of adequateness of classical dissimilarity coefficients is also studied. We underline that a related problem is the popular {\em Multi-Dimensional Scaling}, which maps data points  into a lower-dimensional space and at the same time preserves similarities and dissimilarities between them (see, e.g., \cite{Bent1,Bent2}, \cite{Bui}  and \cite{Okada}).

The paper is structured as follows. Section 2 presents the main preliminary definitions and a  result which characterizes adequateness by using a monotonicity property. Section 3 contains the main results on the characterization of objective functions which are adequate with respect to affine strictly increasing transformations. Section 4 discusses the adequateness of dissimilarity coefficients. Finally, Section 5 illustrates the conclusions and the future lines of research.
	
\section{Notation and preliminary results}

Let  $\Omega$  be some non-empty set,  $n \geq 1$ a natural number and $\mathcal X$  a non-empty set of real-valued  functions  on  $\Omega$.  The  reader  may  recall that  every  function  $X \in {\mathcal X}$ can  be identified with the tuple $(X(\omega ))_{\omega \in \Omega} \in {\mathbb R}^{\Omega}$. In the practice of empirical research $X(\omega )$ is, for every
$X \in {\mathcal X}$, the value of some measurement or judgement that has been taken by $X$ with respect to
$\omega \in \Omega$. Therefore, $\mathcal X$ represents, for every $\omega \in \Omega$, the {\em frequency} of measurements or
judgements that have been taken with respect to $\omega \in \Omega$.

It is more than reasonable to assume that, for every function $X \in {\mathcal X}$ and for every  $\omega \in \Omega$, each value $X(\omega )$ only is known up to some group
$ ({\bf T}_{\omega}, \circ ) $ which consists of transformations (bijective functions) ${ T}_{\omega}: {\mathbb R} \mapsto {\mathbb R}$ (converters) that
represent the particular degree of measurement of the datum observed by
$X(\omega )$ and, on the other hand, that any analysis of measured (observed) data $(X(\omega ))_{\omega \in \Omega}$
for
every $X \in {\mathcal X}$ must be compatible with the sub-group $({\bf T}, \circ )$ of the direct product
$({\bf T}^{\times}, \circ ):=(\displaystyle{\bigtimes_{\omega \in \Omega}}{\bf T}_{\omega } , \circ )$.

In order to now proceed in defining the concept of an objective function $G: {\mathcal X}^n \mapsto \mathbb R$ that is adequate with respect to some sub-group $({\bf T}, \circ )$ of
$({\bf T}^{\times}, \circ ):=(\displaystyle{\bigtimes_{\omega \in \Omega}}{\bf T}_{\omega } , \circ )$ representing the nature of data,  we must assume, of course, $\mathcal X$ to be ${\bf T}$ - {\em closed}, i.e. to contain, with every function  $X=(X(\omega))_{\omega \in \Omega}\in {\mathcal X}$   and every transformation ${\bf \mathcal T}=(T_{\omega})_{\omega \in \Omega} \in {\bf T}$, also the function ${\bf \mathcal T}(X)=(T_{\omega}(X(\omega)))_{\omega \in \Omega}$. Since the analysis of the measured (observed) data shall
be reached by optimizing an objective function $G$, the following definition of an {\em invariance set} is needed (see Subsection 1.7 in \cite{Herd5}).

\begin{definition} \em  Let  $G: {\mathcal X}^n \mapsto \mathbb R$ be an objective function. For every subset  $\mathcal G$ of   ${\mathcal X}^n$ that contains at least two elements, the {\em invariance set} ${\bf I}_{\mathcal G}(G)$  consists of all
bijective functions $S=(S_{\omega})_{\omega \in \Omega} : {\mathbb R}^{\Omega}\mapsto {\mathbb R}^{\Omega}$  that satisfy the implication

\begin{eqnarray*} && G(X_1 , ... , X_n)= \min_{(Y_1 , ..., Y_n)\in {\mathcal G}}G(Y_1 ,...,Y_n) \Rightarrow \nonumber \\ && G(S(X_1),...,S(X_n))=\min_{(S(Y_1) , ..., S(Y_n))\in S({\mathcal G})}G(S(Y_1) ,...,S(Y_n)).\end{eqnarray*}

\end{definition}

 Let, as usual, ${\bf P}({\mathcal X}^n)$ denote the power set of ${\mathcal X}^n$.  Then we set  ${\bf P}^-({\mathcal X}^n):={\bf P}({\mathcal X}^n) \setminus \{{\mathcal G} \subset {\mathcal X}^n \mid \mbox{card}({\mathcal G}) \leq 1\}$.    Now the reader may recall from the above considerations that the group  $({\bf T}, \circ)$ exactly  describes the degree of measurement up to which the
data given by $X=(X(\omega))_{\omega \in \Omega} \in {\mathcal X}$  are known. This means that instead of $X=(X(\omega))_{\omega \in \Omega}$ for
any transformation ${\bf \mathcal T}= (T_{\omega})_{\omega \in \Omega}\in {\bf T}_{\omega}$ also the values  ${\bf \mathcal T}(X)= (T_{\omega}(X(\omega )))_{\omega \in \Omega}$ could have been
measured (observed). This last conclusion implies that in order for $G$ being adequate with respect to $({\bf T}, \circ)$  for every ${\mathcal G} \in {\bf P}^-({\mathcal X}^n)$,  the set ${\bf T}$  must be contained in the invariance set  ${\bf I}_{\mathcal G}(G)$. Indeed, otherwise optimal solutions lose their validity when converting
data from one unit of measurement into another unit of measurement or one rating scale into
another one. This bad consequence, however, means that an analysis of data would depend on
the arbitrarily chosen units of measurement. But, as having stated above, this consequence in
any case has to be avoided in order to at least guarantee a minimum of value of the taken
analysis.

We summarize the previous considerations in the following fundamental definition of an objective function which is {\em adequate} with respect to a transformation group $({\bf T}, \circ )$.

\begin{definition} \em An objective function $G: {\mathcal X}^n \mapsto \mathbb R$ is said to be {\em adequate with respect to a transformation group $({\bf T}, \circ )$} if

\begin{eqnarray*} && G(X_1 , ... , X_n)= \min_{(Y_1 , ..., Y_n)\in {\mathcal G}}G(Y_1 ,...,Y_n) \Rightarrow \nonumber \\ && G({\mathcal T}(X_1),...,{\mathcal T}(X_n)) =\min_{{\mathcal T}(Y_1) , ..., {\mathcal T}(Y_n)\, s.t. \, {(Y_1 , ..., Y_n)\in {\mathcal G}}}G({\mathcal T}(Y_1) ,...,{\mathcal T}(Y_n))  \nonumber \\ &&\mbox{ for all ${\mathcal T} \in {\bf T}$ and all sets ${\mathcal G} \subset {\mathcal X}^n$}.\end{eqnarray*}

\end{definition}

Consequently, we believe that one of the most fundamental challenges in multivariate data analysis is characterizing for any given sub-group $({\bf T}, \circ)$ of $(\displaystyle{\bigtimes_{\omega \in \Omega}} {\bf T}_{\omega}, \circ)$ all objective functions that are adequate with respect to $({\bf T}, \circ)$.

Following the idea of the proof of Proposition A1 in  \cite{Herd5} or the idea of the proof of Proposition 2.2 in  \cite{Herd4}, the next proposition can be easily proven.

\begin{proposition} \label{propocar} \  $G: {\mathcal X}^n \mapsto \mathbb R$ be an objective function. Then the following conditions are equivalent:

\begin{enumerate} \item  $G $ is adequate with respect to the transformation group $({\bf T}, \circ)$;\\

 \item $G$ satisfies, for all tuples  $(X_1 , ..., X_n)\in {\mathcal X}^n$ and $(Y_1 , ..., Y_n)\in {\mathcal X}^n$  and all transformations  ${\bf \mathcal T} \in {\bf T}$, the following conditions of monotonicity:

   \begin{eqnarray*}G(X_1 , ..., X_n) < G(Y_1 , ..., Y_n) \Rightarrow \\G({\bf \mathcal T}(X_1) , ..., {\bf \mathcal T}(X_n)) < G({\bf \mathcal T}(Y_1) , ..., {\bf \mathcal T}(Y_n)).\end{eqnarray*}

\end{enumerate}

\end{proposition}

\begin{proof} 1 $\Rightarrow$ 2. Consider any two tuples $(X_1 , ..., X_n), (Y_1 , ..., Y_n) \in {\mathcal X}^n$, and define ${\mathcal G}:=\{(X_1 , ..., X_n), (Y_1 , ..., Y_n)\}$. Since, in particular, $G$ belongs to the invariance set ${\bf I}_{\mathcal G}(G)$ due to the fact that $G $ is adequate with respect to $({\bf T},   \circ)$, we have that $G(X_1 , ..., X_n) < G(Y_1 , ..., Y_n)$, which implies that  $\min \{G(X_1 , ..., X_n), G(Y_1 , ..., Y_n)\}= G(X_1 , ..., X_n)$. Therefore,
$$\min\{G({\bf \mathcal T}(X_1) , ..., {\bf \mathcal T}(X_n)), G({\bf \mathcal T}(Y_1) , ..., {\bf \mathcal T}(Y_n))\} = G({\bf \mathcal T}(X_1) , ..., {\bf \mathcal T}(X_n))$$ 
for every transformation  ${\bf \mathcal T} \in {\bf T}$, and clearly this last equality entails $$G({\bf \mathcal T}(X_1) , ..., {\bf \mathcal T}(X_n)) < G({\bf \mathcal T}(Y_1) , ..., {\bf \mathcal T}(Y_n)).$$

2 $\Rightarrow$ 1. Consider any set  ${\mathcal G} \subset {\mathcal X}^n$ and assume that \begin{eqnarray*}  (G(X_1 , ... , X_n) = \min_{(Y_1 , ..., Y_n)\in {\mathcal G}}G(Y_1 ,...,Y_n)) \,\, and \nonumber  \\  (G(X'_1 , ... , X'_n) \neq  \min_{G(Y_1 , ..., Y_n)\in {\mathcal G}}G(Y_1 ,...,Y_n) \,\,\, \forall (X'_1 , ... , X'_n) \neq (X_1 , ... , X_n)) \Leftrightarrow \\ G(X_1 , ... , X_n) < G(Y_1 ,...,Y_n)\,\,\,\forall (Y_1 , ..., Y_n)\in {\mathcal G}.\end{eqnarray*} Further, consider any transformation  ${\bf \mathcal T} \in {\bf T}$. Then we have that, by condition 2, $$ G({\mathcal T}(X_1),...,{\mathcal T}(X_n)) <  G({\mathcal T}(Y_1) ,...,{\mathcal T}(Y_n))\,\,\,\forall (Y_1 , ..., Y_n)\in {\mathcal G},$$ which in turn is equivalent to \begin{eqnarray*}  (G({\mathcal T}(X_1),...,{\mathcal T}(X_n)) =\min_{\mathcal T(Y_1) , ..., \mathcal T(Y_n) \, s.t.\, {(Y_1 , ..., Y_n)\in {\mathcal G}}}G({\mathcal T}(Y_1) ,...,{\mathcal T}(Y_n))) \,\, and \\ (G({\mathcal T}(X'_1) , ... , {\mathcal T}(X'_n)) \neq  \min_{{\mathcal T}(Y_1) , ..., {\mathcal T}(Y_n)\, s.t.\, {(Y_1 , ..., Y_n)\in {\mathcal G}}}G({\mathcal T}(Y_1) ,...,\mathcal T(Y_n))\\  \forall (X'_1 , ... , X'_n) \neq (X_1 , ... , X_n))\end{eqnarray*} This consideration completes the proof.
\end{proof}

\begin{remark} Since  $({\bf T}, \circ)$  is a group, the condition 2 of the above Proposition \ref{propocar}, of course, can be replaced by the equivalence $$G(X_1 , ..., X_n) \leq G(Y_1 , ..., Y_n) \Leftrightarrow G({\bf \mathcal T}(X_1) , ..., {\bf \mathcal T}(X_n)) \leq G({\bf \mathcal T}(Y_1) , ..., {\bf \mathcal T}(Y_n))\, .$$ \end{remark} 	

Finally, the reader may verify that the adequateness of $G$  with respect to $({\bf T}, \circ)$   implies that for all sets ${\mathcal G} \in {\bf P}^-({\mathcal X}^n)$  the invariance set ${\bf I}_{\mathcal G}(G)$  actually is a group $({\bf I}_{\mathcal G}(G), \circ )$.

\section{Adequateness theorems for ordinal or interval scaled data}

In this section we focus on the proofs of the main theorems of this paper.
This means that the consequences of these theorems for the practice of empirical research
mainly will be discussed in the conclusion of the paper.
	
	In order to carefully prepare these theorems, let for every $\omega \in \Omega$  some group $({\bf T}_{\omega}, \circ )$  of affine transformations $T_{\omega}: {\mathbb R} \mapsto {\mathbb R}$  be fixed chosen. Concentrating at first on interval scaled
data, we assume that, for every $\omega \in \Omega$,  each group $({\bf T}_{\omega}, \circ )$ consists of all affine transformations $T_{\omega}(r)=a_{\omega} \cdot r + t_{\omega}$ ($r \in {\mathbb R}$), which are order-automorphisms, i.e. $a_{\omega} , t_{\omega} \in {\mathbb R}$ and $a_{\omega} > 0$. Then,
we still notice that in this section no group $({\bf T}_{\omega}, \circ )$ merely consists of bounded order-automorphisms
on $\mathbb R$. Consequently, in the first paragraph of this section we do not require
$\mathcal X$ to only consist of bounded real-valued functions.

\begin{definition} \em We define the sub-group
$({\bf T}_{\bf c}, \circ )$ of $({\bf T}^{\times}, \circ )= \left(\displaystyle{\bigtimes_{\omega \in \Omega}} {\bf T}_{\omega}\right)$ as follows:

$({\bf T}_{\bf c}, \circ )$ consists of all tuples $(T_{\omega})_{\omega \in \Omega} \in ({\bf T}^{\times}, \circ )= \left( \displaystyle{\bigtimes_{\omega \in \Omega}} {\bf T}_{\omega}\right)$
such that $T_{\omega} = T_{\omega '}$, i.e. $a_{\omega} = a_{\omega '}$ and $t_{\omega} = t_{\omega '}$ for all $\omega \in \Omega$ and all $\omega ' \in \Omega$.

\end{definition}

We now focus on the problem of
characterizing all objective functions $G:{\mathcal X}^n \mapsto {\mathbb R}$ that
are adequate with respect to $({\bf T}_{\bf c}, \circ )$.

We recall that the {\em zero-function} ${\bf 0}:\Omega \mapsto {\mathbb R}$ maps every $\omega \in \Omega$ on $0$. In the remainder of this paper we assume {\bf 0} to be contained in $\mathcal X$. Replacing $G:{\mathcal X}^n \mapsto {\mathbb R}$ by
$G - G({\bf 0},...,{\bf 0})$ we may assume, without loss of generality, that
$G({\bf 0},...,{\bf 0})=0$.

Given any objective function $G:{\mathcal X}^n \mapsto {\mathbb R}$ that
is adequate with respect to $({\bf T}_{\bf c}, \circ )$, we consider, for every $\omega \in \Omega$,  the function $X_{\omega} \in {\mathcal X}$  that is defined by setting

  \[ X_{\omega}(\omega') :=
\left\{
\begin{array}{ll} 1 & \mbox{if $\omega = \omega'$}
\\ 0 & \mbox{if $\omega \neq \omega'$}
  \end{array} \right.
\]

\noindent for all $\omega \in \Omega$.

\begin{definition} \em An objective function  $G:{\mathcal X}^n \mapsto {\mathbb R}$ is said to have a {\em relevant components} $\omega \in \Omega$ if it happens that  $$G(X_{\omega},...,X_{\omega}) \neq G({\bf 0},...,{\bf 0}).$$
\end{definition}

 Of course,  the existence of relevant components always is satisfied by objective functions $G:{\mathcal X}^n \mapsto {\mathbb R}$  that are nontrivial (i.e., not identically equal to zero). At this stage, we mention that the concept of a relevant component in the literature appears in various contexts which may widely differ (cf. \cite{Mart},  \cite{Hertz} and \cite{Bar}). In the framework of our considerations, the concept of a component that is relevant with respect to $G$ may depend on the particular type of objective function that is represented by $G$. Indeed, within the environment of dissimilarity coefficients or stress measures (stress functions), a somewhat different concept of a component being relevant with respect to $G$ has to be introduced.

We are now ready to present the main result of this paper, which provides a characterization of all objective functions which are adequate with respect to the transformation group of all affine strictly increasing real-valued bijections of the real line into itself. 

We just recall that a function $H : {\mathbb R} \mapsto {\mathbb R}^{\geq 0}$ is said to be {\em strictly isotone} if, for all $x,y \in {\mathbb R}$, $x < y$ implies that $H(x) < H(y)$, where $\leq$ and $<$ are the {\em natural total order} on ${\mathbb R}$ and its {\em strict part}, respectively.
 
\begin{theorem} \label{Theochar}
Let $G:{\mathcal X}^n \mapsto {\mathbb R}^{\geq 0}$ be a nontrivial objective function. Then the following conditions are equivalent: \begin{enumerate}

\item  There exists a unique relevant component $\gamma$ and some strictly isotone function
$H : {\mathbb R} \mapsto {\mathbb R}^{\geq 0}$, such that $H(r) > 0$ for every $r > 0$ and
\[ G(X_1 , ... , X_n )=
\left\{
\begin{array}{ll} H(X_1 (\gamma )) & \mbox{if $X_1 (\gamma )= ... = X_n (\gamma )$}
\\ 0 & \mbox{otherwise}
  \end{array} \right.
\]

for all tuples $(X_1 , ... , X_n )\in {\mathcal X}^n$;\\

\item $G$ is adequate with respect to $({\bf T}_{\bf c}, \circ )$.

\end{enumerate}

\end{theorem}

\begin{proof}

1 $\Rightarrow$ 2. Let  $G:{\mathcal X}^n \mapsto {\mathbb R}^{\geq 0}$ be defined by some strictly isotone function $H : {\mathbb R} \mapsto {\mathbb R}^{\geq 0}$ in the way that is described in the theorem. Then the definition of the group $({\bf T}_{\bf c}, \circ )$ implies that $G$ is adequate
with respect to $({\bf T}_{\bf c}, \circ )$ as a consequence of the characterization of adequate objective functions in Proposition \ref{propocar}. Indeed, consider any two tuples  $(X_1 , ..., X_n),\,(Y_1 , ..., Y_n)\in {\mathcal X}^n$  such that $G(X_1 , ..., X_n) < G(Y_1 , ..., Y_n)$. Clearly, there are four different cases with respect to the possible values of $G$, according to condition 1. Since all the cases are treated analogously, we shall limit ourselves to the case when $G(X_1 , ..., X_n)= H(X_1 (\gamma ))$ and $G(Y_1 , ..., Y_n)= H(Y_1 (\gamma ))$. Since $H$ is strictly isotone,  we have that $H(X_1 (\gamma )) < H(Y_1 (\gamma ))$ implies that $X_1 (\gamma ) < Y_1 (\gamma )$. Therefore, $a X_1 (\gamma ) + t < a y_1 (\gamma ) + t$ gives  $H(a X_1 (\gamma ) + t) < H(a y_1 (\gamma ) + t)$, or equivalently $G({\bf \mathcal T}(X_1) , ..., {\bf \mathcal T}(X_n)) < G({\bf \mathcal T}(Y_1) , ..., {\bf \mathcal T}(Y_n))$.

2 $\Rightarrow$ 1. Let  $G:{\mathcal X}^n \mapsto {\mathbb R}^{\geq 0}$ be a nontrivial objective function which is adequate with respect to $({\bf T}_{\bf c}, \circ )$. It has been already observed that nontriviality of $G$ guarantees the existence of a relevant conmponent $\gamma$. Let us show that $\gamma$ is the only relevant component of $G$. This means that we must
prove that the assumption $\omega \in \Omega \setminus \{\gamma\}$ to be another relevant component of $G$ contradicts the
adequateness of $G$ with respect to $({\bf T}_{\bf c}, \circ )$. Let us assume, in contrast, that $\gamma$ and $\omega \in \Omega \setminus \{\gamma\}$
are both relevant components of $G$. Then, from our assumptions, both values $G(X_{\gamma}, ... , X_{\gamma })$ as well as $G(X_{\omega}, ... , X_{\omega })$ are both strictly greater $0$. In order to show that this is contradictory, we proceed by illustrating the following three steps.\\ \noindent
In the first step we define for $\gamma \in \Omega$ as well as for $\omega \in \Omega$ functions $F_{\gamma}: {\mathbb R}^{\geq 0} \mapsto {\mathbb R}^{\geq 0}$ and
$F_{\omega}: {\mathbb R}^{\geq 0} \mapsto {\mathbb R}^{\geq 0}$  by setting $F_{\gamma}(t):=G(t \cdot X_{\gamma}, ... , t \cdot X_{\gamma })$ for every $t \in {\mathbb R}^{\geq 0}$ and $F_{\omega}(t):=G(t \cdot X_{\omega}, ... , t \cdot X_{\omega })$ for every $t \in {\mathbb R}^{\geq 0}$. Now we show that the function $F_{\gamma}$ as well as the
function $F_{\omega}$ is strictly increasing. Of course, it is enough to verify that $F_{\gamma}$ is strictly increasing.
Then replacing $F_{\gamma}$ by $F_{\omega}$ it also follows that $F_{\omega}$ is strictly increasing. In order to show that
$F_{\gamma}$ is strictly increasing we consider two non-negative real numbers $r < s$ and assume, in contrast,
that $F_{\gamma}(s)=G(s \cdot X_{\gamma}, ... , s \cdot X_{\gamma}) \leq G(r \cdot X_{\gamma}, ... , r \cdot X_{\gamma}) =F_{\gamma}(r)$. The adequateness of $G$ with respect to $({\bf T}_{\bf c}, \circ )$, then allows us to conclude that $G((s- r) \cdot X_{\gamma}, ... , (s- r) \cdot X_{\gamma}) \leq 0$. Hence,
one more application of the adequateness of $G$ with respect to $({\bf T}_{\bf c}, \circ )$ implies that
$G(\frac{1}{s-r}\cdot (s- r) \cdot X_{\gamma}, ... , \frac{1}{s-r}\cdot(s- r) \cdot X_{\gamma}) = G(X_{\gamma}, ... , X_{\gamma})\leq 0$. But this inequality contradicts our assumption $G(X_{\gamma}, ... , X_{\gamma})$ to be strictly greater than $0$. For later use we still notice
that the monotonicity of $F_{\gamma}$ as well as the monotonicity of $F_{\omega}$ implies that $F_{\gamma}$  and $F_{\omega}$ as
well have at most countably many points of discontinuity.\\ \noindent
In the second step we assume that there exist positive reals $r$ and $s$ such that $F_{\gamma}(r)=G(r \cdot X_{\gamma}, ... , r \cdot X_{\gamma})=F_{\omega}(s)=G(s \cdot X_{\omega}, ... , s \cdot X_{\omega})$. Then it follows from the adequateness of
 $G$ with respect to $({\bf T}_{\bf c}, \circ )$ that $G(\frac{1}{r}\cdot r \cdot X_{\gamma}, ... , \frac{1}{r}\cdot r \cdot X_{\gamma})= G( X_{\gamma}, ... , X_{\gamma})=F_{\gamma}(1)=G(\frac{1}{s}\cdot s \cdot X_{\omega}, ... , \frac{1}{s}\cdot s \cdot X_{\omega})= G( X_{\omega}, ... ,  X_{\omega})=F_{\omega}(1)$. One more application of the adequateness of
 $G$ with respect to $({\bf T}_{\bf c}, \circ )$ then implies that $G(2 \cdot X_{\gamma},...,2 \cdot X_{\gamma})= G( X_{\omega},..., X_{\omega})$. Hence, we
may conclude that $G(2 \cdot X_{\gamma},...,2 \cdot X_{\gamma})=G(X_{\gamma},..., X_{\gamma})=F_{\gamma}(1)$ which contradicts the
conclusion of the first step where it has been proved that $F_{\gamma}:{\mathbb R}^{\geq 0} \mapsto {\mathbb R}^{\geq 0}$ is strictly increasing. Hence, there cannot exist positive reals $r$ and $s$ such that $ F_{\gamma} (r)=G(r \cdot X_{\gamma},...,r \cdot X_{\gamma}) = F_{\omega} (s)=G(s \cdot X_{\omega},...,s \cdot X_{\omega})$.\\ \noindent  In order to, finally, finish in the third step we prove the incompatibility of the existence of
another relevant component $\omega$ with the adequateness of $G$ with respect to $({\bf T}_{\bf c}, \circ )$. We verify
at first that $G(X_{\gamma},..., X_{\gamma}) < G(X_{\gamma} + X_{\omega},..., X_{\gamma}+X_{\omega})$. Indeed, if we assume, in contrast, that
$ G(X_{\gamma} + X_{\omega},..., X_{\gamma}+X_{\omega}) \leq G(X_{\omega},..., X_{\omega})$ then the adequateness of $G$ with respect to $({\bf T}_{\bf c}, \circ )$
implies that  $ G(X_{\gamma} + X_{\omega} - X_{\gamma},..., X_{\gamma}+X_{\omega}-X_{\gamma}) = G(X_{\omega},..., X_{\omega}) \leq G({\bf 0}, ... , {\bf 0})$ which contradicts
our assumption on $G(X_{\omega},..., X_{\omega})$. In the remainder of the proof the above
considerations will be abbreviated by (*). In order to now proceed we still notice that the
arguments that have been applied in the second step allow us to conclude the non-existence of positive reals $r$ and $s$ such that $F_{\gamma}(r) = G(r \cdot X_{\gamma}, ... , r \cdot X_{\gamma}) = F_{\omega}(s) = G(s \cdot X_{\omega}, ... , s \cdot X_{\omega})$. One more application of those arguments, therefore, implies that we may assume (without loss of generality) that for every
fixed chosen positive real $s$ and for all positive reals $r$, the strong inequality $ F_{\omega}(r) = G(r \cdot X_{\gamma}, ... , r \cdot X_{\gamma}) < F_{\omega}(s) = G(s \cdot X_{\omega}, ... , s \cdot X_{\omega})$  holds. It, thus, follows that 

$$\displaystyle{\sup_{r \in {\mathbb R}^{\geq 0}}} F_{\gamma}(r) = \displaystyle{\sup_{r \in {\mathbb R}^{\geq 0}}} G(r \cdot X_{\gamma}, ... , r \cdot X_{\gamma}) \leq \displaystyle{\inf_{s \in {\mathbb R}^{\geq 0}}} F_{\omega}(s)= \displaystyle{\inf_{s \in {\mathbb R}^{\geq 0}}} G(s \cdot X_{\omega}, ... , s \cdot X_{\omega})\, .$$ 

We abbreviate these considerations
by (**). Let now some fixed positive real $a$ and some strictly decreasing sequence
$(b_k)_{1 \leq k < \infty}$
be arbitrarily chosen in such a way that $a < b_k$ for all $1 \leq k < \infty$ and $\displaystyle{\lim_{k \rightarrow \infty}}b_k = a$. Then
a routine argument that is based upon the adequateness of $G$ with respect to $({\bf T}_{\bf c}, \circ )$ allows to conclude, by applying the considerations (*) and (**), that for every natural number $1 \leq k < \infty$, for every natural number $k < t < \infty$ and for every positive real $c \geq  a$, the equations
and inequalities $F_{\omega}(a) = G(a \cdot X_{\omega}, ... , a \cdot X_{\omega}) < G(c \cdot X_{\gamma} + t \cdot X_{\omega}, ... , c \cdot X_{\gamma} + t \cdot X_{\omega}) < G(b_k \cdot X_{\omega}, ... ,b_k \cdot X_{\omega} )= F_{\omega}(b_k )$ hold. Hence, the inequalities and equation 
\begin{equation*}
\begin{split}
F_{\omega}(a) &< {\sup_{c \in {\mathbb R}^{\geq 0}}} {\lim_{t \rightarrow \infty}} G(c \cdot X_{\gamma} + t \cdot X_{\omega}, ... , c \cdot X_{\gamma} + t \cdot X_{\omega}) \\
&\leq {\lim_{k \rightarrow \infty}} G(b_k \cdot X_{\omega} , ... , b_k \cdot  X_{\omega}) \\
&= {\lim_{k \rightarrow \infty}} F_{\omega}(b_k )
\end{split}
\end{equation*} 
hold.
Therefore, $F_{\omega}$ is discontinuous at $a$. Since $a$ has been chosen arbitrarily it follows that $F_{\omega}$ is
discontinuous at every positive real. But in the first step it has been shown that $F_{\gamma}$ only has
countably many points of discontinuity. This contradiction finishes the third step and, therefore,
the proof of the incompatibility of the assumed additional existence of the relevant component
$\omega$ with the adequateness of $G$ with respect to $({\bf T}_{\bf c}, \circ )$.\\ \noindent
In order to now end the proof of the theorem we may conclude as follows: since
$G(X_{\omega '}, ... , X_{\omega '}) = 0$ if and only if $\omega ' \in \Omega \setminus \{\gamma \}$, it follows that
\[ G(X_1 , ... , X_n )=
\left\{
\begin{array}{ll} G(X_1 (\gamma )\cdot X_{\gamma}, ... ,X_1 (\gamma )\cdot X_{\gamma} ) & \mbox{if $X_1 (\gamma )= ... = X_n (\gamma )$}
\\ 0 & \mbox{otherwise}
  \end{array} \right.
\]
for all tuples $(X_1 , ... , X_n ) \in {\mathcal X}^n$. Hence, we may define the desired function $H : {\mathbb R} \mapsto {\mathbb R}$ by
setting $H(r) := G(r \cdot X_{\gamma}, ... , r \cdot X_{\gamma})$ for all $r \in {\mathbb R}$. Since $\gamma$  is a relevant component, we may
conclude that $H(r) \in {\mathbb R} \setminus \{0\}$  for all $r \in {\mathbb R} \setminus \{0\}$. In addition, it follows from the first step of
the above proof that $H_{\mid {\mathbb R}^{\geq 0}}=F_{\gamma}$ is strictly isotone. But, an analysis of the arguments that have
been used in the first step of the above proof implies that these arguments apply for arbitrary
reals $r < s$. Hence, $H : {\mathbb R} \mapsto {\mathbb R}^{\geq 0}$ is strictly isotone.  This observation
completes the proof of the theorem. \end{proof}

Let us assume, for the moment, $\mathcal X$ to only consist of bounded functions.
Then ${\mathcal X}$ in a standard way is endowed with its natural topology $t_{nat}$ that generalizes the natural topology defined on the reals.
A basis of $t_{nat}$ is given by considering for all $X\in {\mathcal X}$ and all positive reals $\varepsilon$ the open neighborhoods $U_{\varepsilon}(X):=\{Y \in {\mathcal X} : \, \displaystyle{\sup_{\omega \in \Omega}}|X(\omega)-Y(\omega)|<\varepsilon\}$ of $X$.
The cross product ${\mathcal X} ^n$ of ${\mathcal X}$ then is endowed with the corresponding product topology $t^n _{nat}$.
Now the question arises if Theorem \ref{Theochar} can be strengthened by requiring $G$ to not only being adequate with respect to $(\bf{T}_c,\circ)$ but also being continuous with respect to the topology $t^n _{nat}$ on ${\mathcal X} ^n$ and the topology $t_{nat}$ on $\mathbb{R}$.
Due to the obvious importance of continuity for the optimization on compact sets, we state the following theorem which represents the continuous version of Theorem \ref{Theochar}.

\begin{theorem} \label{Theocharcont}
Let $G:{\mathcal X}^n \mapsto {\mathbb R}^{\geq 0}$ be a nontrivial objective function. Then the following conditions are equivalent: \begin{enumerate}

\item  There exists a unique relevant component $\gamma$ and some strictly isotone continuous function
$H : {\mathbb R} \mapsto {\mathbb R}^{\geq 0}$ such that $H(r) > 0$ for every $r > 0$ and
\[ G(X_1 , ... , X_n )=
\left\{
\begin{array}{ll} H(X_1 (\gamma )) & \mbox{if $X_1 (\gamma )= ... = X_n (\gamma )$}
\\ 0 & \mbox{otherwise}
  \end{array} \right.
\]

for all tuples $(X_1 , ... , X_n )\in {\mathcal X}^n$;\\

\item  $G:({\mathcal X} ^n,t^n _{nat})\mapsto ({\mathbb R}^{\geq 0}, t_{nat})$ is continuous and adequate with respect to $({\bf T}_{\bf c}, \circ )$.

\end{enumerate}

\end{theorem}

\begin{proof} 1 $\Rightarrow$ 2. Please refer to the corresponding proof of the implication ``1 $\Rightarrow$ 2'' in Theorem \ref{Theochar}, and consider that, obviously, now $G$ is continuous since $H$ is continuous.

2 $\Rightarrow$ 1. Needless to say, the construction in the corresponding part of the proof of Theorem \ref{Theochar} is still valid. Now,  the function $H : {\mathbb R} \mapsto {\mathbb R}^{\geq 0}$ defined by
setting, for all $r \in {\mathbb R}^{\geq 0}$, $H(r) := G(r \cdot X_{\gamma}, ... , r \cdot X_{\gamma})$ is continuous as a consequence of the fact that $G$ is continuous. This consideration finishes the proof.
\end{proof}

\begin{remark} Theorem \ref{Theochar} and Theorem \ref{Theocharcont} as well are theorems of ``Arrow type'' (cf. Arrow \cite{Arr}) meaning that the value of $G$ depends on the value of exactly one $\omega \in \Omega$  which, therefore, is the ``dictator'' of $G$.
This means that requiring non-dictatorship implies that there exists no function $G:{\mathcal X} ^n\mapsto\mathbb{R}$ which is adequate with respect to $(\bf{T}_c,\circ)$.
\end{remark}

Of course, Theorem \ref{Theochar} and Theorem \ref{Theocharcont} respectively cannot be strengthened by concentrating on ordinal scaled data.
Both theorems are extremely restrictive.
They imply the non-existence of suitable objective functions $G:{\mathcal X}^n\mapsto \mathbb{R}$ that are adequate with respect to $(\bf{T}_c,\circ)$ and can be applied in the practice of empirical research (cf. also the discussion of Theorem \ref{Theochar} and Theorem \ref{Theocharcont} in the conclusion of this paper).
The only possibility to avoid this consequence is to rigorously restrict the number of transformations allowed.

\section{Adequateness of dissimilarity coefficients}

We now concentrate on dissimilarity coefficients as frequently applied in cluster analysis (see, e.g., \cite{Hei}, \cite{Hei1}). We just mention the increasing importance of measures of functional similarity in Ecological Modelling (see e.g. \cite{Ric}).

\begin{definition} \em If $\Omega$ is finite then the following dissimilarity coefficients $d:{\mathcal X} \times {\mathcal X} \mapsto \mathbb{R}^{\geq 0}$ are well known:

\begin{enumerate} \item  $d(X,Y):=\displaystyle{\sum_{\omega \in \Omega}} \vert X(\omega)-Y(\omega)\vert$ for all pairs $(X,Y) \in {\mathcal X} \times {\mathcal X}$\linebreak  ({\em city block metric}),
\item  $d(X,Y):=(\displaystyle{\sum_{\omega \in \Omega}} \vert X(\omega)-Y(\omega)\vert ^2)^{1/2}$ for all pairs $(X,Y) \in {\mathcal X} \times {\mathcal X}$\linebreak  ({\em Euclidean metric}).

\end{enumerate}
\end{definition}
Both dissimilarity coefficients belong to the so-called $L_r$--distances that for every fixed natural number $1\leq m<\infty$ are defined by setting, for all pairs $(X,Y) \in {\mathcal X} \times {\mathcal X}$,

$$d(X,Y):=(\displaystyle{\sum_{\omega \in \Omega}} \vert X(\omega)-Y(\omega)\vert ^m)^{1/m} .$$

Because of these examples, let us define a transformation group of affine functions which guarantees the adequateness of the previously defined dissimilarity coefficients.

\begin{definition} \label{defaff} \em Consider, for every $\omega \in \Omega$, the group $(\bf{T}_{\omega}, \circ)$ that consists of all affine functions $T_{\omega}:\mathbb{R}\mapsto \mathbb{R}$ for which there exists some real $a_{\omega}$ and some real $t_{\omega}$ that is different from 0 whenever $a_{\omega} =0$ such that $T_{\omega}(r)=a_{\omega}\cdot r + t_{\omega}$ for all $r\in \mathbb{R}$.
Then we choose the subset $\bf{T}_r ^{\times}$ of $\displaystyle{\bigtimes_{\omega \in \Omega}} \bf{T}_{\omega}$ that contains all tuples (transformations) $\bf{T}=(T_{\omega})_{\omega \in \Omega}$ of affine functions $T_{\omega}$ on $\mathbb{R}$ and which satisfy the additional (restrictive) property that for all $\omega \in \Omega$ and all $\omega ' \in \Omega$ the equations $T_{\omega} (r)=a_{\omega} \cdot r+t_{\omega}$ and $T_{\omega '}(r)=a_{\omega '}\cdot r+t_{\omega '}$ imply that the absolute values $\vert a_{\omega} \vert$ and $\vert a_{\omega '} \vert$ coincide.

\end{definition}

We are now ready to prove the following proposition.

\begin{proposition} The city block metric and Euclidean metric dissimilarity coefficients $d:{\mathcal X} \times {\mathcal X} \mapsto {\mathbb R}^{\geq 0}$are adequate with respect to
$(\bf{T}_r ^{\times}, \circ )$.
\end{proposition}

\begin{proof} We shall use the characterization presented in Proposition \ref{propocar}. We will only prove the adequateness of the city block metric. The proof of the adequateness of the Euclidean norm is analogous and it is omitted for the sake of brevity. Consider $X,Y,Z,W$ such that $$d(X,Y) =  \displaystyle{\sum_{\omega \in \Omega}} \vert X(\omega)-Y(\omega)\vert < \displaystyle{\sum_{\omega \in \Omega}} \vert Z(\omega)-W(\omega)\vert = d(Z,W).$$ Then the city block dissimilarity coefficient is adequate with respect to $(\bf{T}_r ^{\times}, \circ )$ provided that $$ \displaystyle{\sum_{\omega \in \Omega}} \vert a_{\omega}X(\omega) + t_{\omega}-a_{\omega}Y(\omega) - t_{\omega}\vert < \displaystyle{\sum_{\omega \in \Omega}} \vert  a_{\omega}Z(\omega) + t_{\omega}- a_{\omega}W(\omega)- t_{\omega}\vert ,$$ where $a_{\omega}$ and $t_{\omega}$ respect the conditions of Definition \ref{defaff}. But the previous inequality is equivalent to the following one: $$\displaystyle{\sum_{\omega \in \Omega}} \vert a_{\omega} \vert \vert X(\omega)  - Y(\omega) \vert <  \displaystyle{\sum_{\omega \in \Omega}} \vert a_{\omega} \vert \vert Z(\omega)  - W(\omega) \vert, $$ which is satisfied by hypothesis. This consideration completes the proof.
\end{proof}

\section{Conclusion}

The results that have been presented in this paper are restrictive in nature. It seems that most objective functions that appear in various classical fields of applications, i.e. cluster analysis, factor analysis, (linear) structural equation modeling, (linear) regression, (non-metric) multidimensional scaling, choice theory (social welfare functions) or utility theory are not really appropriate for analyzing data that may be observed in practice because of their lack to being adequate. Hence, on the one hand, the question arises if adequateness perhaps is a criterion that is too strong and, therefore, not really adapted to situations that appear in practice. On the other hand, however, adequateness is not an artificial concept. Indeed, Theorem \ref{Theochar} and Theorem \ref{Theocharcont} suggest that results obtained by applying standardized methods could lead to misinterpretations of data.\\ \indent
Since these considerations, the results presented in the third section of this paper need further reflection. Specifically, Theorem \ref{Theochar} states that there cannot exist any objective function that really is adapted to the nature of the data and, therefore, appropriate for analyzing all types of interval scaled data.

We underline the fact that those results remain valid for any type of data analysis that, finally, is based upon comparison of data. In opinion of the authors, Theorem \ref{Theochar} and Theorem \ref{Theocharcont}, thus, highlight that the concept of adequateness considered in this paper is the right concept for checking if an objective function really is appropriate for analyzing those kind of data. 

However, as a matter of fact, the consideration of all affine order-automorphisms of the real line into itself as possible transformations seems to be too demanding and  leads to a restrictive characterization of adequateness. For that reason, in a future paper we are intended to limit, in some sense, the transformations that are allowed or to consider different kind of transformations in order to arrive at less restrictive characterizations.

In addition, as regards Section 4, we are convinced that it is possible, following the lines of the proof of Theorem \ref{Theochar}, to characterize in a general setting the adequateness of dissimilarity indexes with respect to the transformations considered accordingly. If this possibility will be effective, the corresponding results will be presented in a future paper.

\end{document}